%% file: contractionfactors.tex
\pdfoutput=1
\documentclass[amsmath,secnumarabic,floatfix,amssymb,nofootinbib,nobibnotes,letterpaper,11pt,tightenlines]{revtex4}

\usepackage[T1]{fontenc}
\usepackage{geometry}
\usepackage{latexsym}
\usepackage{amsmath}
\usepackage{amssymb}
\usepackage{amscd}
\usepackage{amsthm}
\usepackage{mathtools}
\usepackage{cancel}
\usepackage[dvipsnames]{xcolor}

\usepackage{newtxtext,newtxmath}

\usepackage{graphicx}
\usepackage[percent]{overpic}
\usepackage{units}

\usepackage{tikz-cd}
\usepackage{algorithm}
\usepackage{algorithmicx}
\usepackage{algpseudocode}

\usepackage{enumitem}
\setlist{nosep}
\def\figdir{Pictures/}
\graphicspath{\figdir}

\usepackage{aliascnt} %alias counter
\newcommand{\newautotheorem}[3] %{environment name}{counter}{displayed name} - Sets the autorefname of enronment so that \autoref from the \hyperref-Packet works correctly.
{
\newaliascnt{#1}{#2}
\newtheorem{#1}[#1]{#3}
\aliascntresetthe{#1}
\expandafter\def\csname #1autorefname\endcsname{%
#3%
}%
}

\newtheorem{theorem}{Theorem}

\newautotheorem{lemma}{theorem}{Lemma}
\newautotheorem{proposition}{theorem}{Proposition}
\newautotheorem{corollary}{theorem}{Corollary}

\theoremstyle{definition}
\newautotheorem{definition}{theorem}{Definition}

\newautotheorem{conjecture}{theorem}{Conjecture}
\newautotheorem{remark}{theorem}{Remark}
\newautotheorem{claim}{theorem}{Claim}

\renewcommand{\Vec}{\operatorname{vec}}

\input{Macros.tex}

\usepackage[
	backref  = false%
	,pagebackref = false%
	,bookmarks  = true%
	,bookmarksdepth=3%
]{hyperref}
\usepackage[all]{hypcap} 

\begin{document}
\title{An exact formula for the contraction factor of \\
a subdivided Gaussian topological polymer}
\author{Jason Cantarella}
\altaffiliation{Mathematics Department, University of Georgia, Athens, GA, USA}
\noaffiliation
\author{Tetsuo Deguchi}
\altaffiliation{Department of Physics, Ochanomizu University, Bunkyo-ku, Tokyo 112--8610, Japan}
\noaffiliation
\author{Clayton Shonkwiler}
\altaffiliation{Department of Mathematics, Colorado State University, Fort Collins, CO, USA}
\noaffiliation
\author{Erica Uehara}
\altaffiliation{Physical Statistics Laboratory, Kyoto University, Sakyo-ku, Kyoto 606--8501, Japan}
\noaffiliation

%\subjclass[2020]{05C09,92E10,82D60} % Graphical indices: Wiener index, Zagreb index, etc..

\begin{abstract}
We consider the radius of gyration of a Gaussian topological polymer $\graphG$ formed by subdividing a graph $\graphG'$ of arbitrary topology (for instance, branched or multicyclic). We give a new exact formula for the expected radius of gyration and contraction factor of $\graphG$ in terms of the number of subdivisions  of each edge of $\graphG'$ and a new weighted Kirchhoff index for $\graphG'$. The formula explains and extends previous results for the contraction factor and Kirchhoff index of subdivided graphs. 
\end{abstract}
\date{\today}
\maketitle

Polymers with new topological structures have recently attracted a great deal of attention, as new syntheses of graph-shaped and multicyclic structures have been accomplished~\cite{Tezuka:2017gh}. We consider the statistical properties of shapes of these polymers using Gaussian ideal chains, as in  ~\cite{James1943,James:1947hp,FloryPaulJ1969Smoc}. 

The ensemble average $\left< \Rog \right>$ of the radius of gyration of the embedding is a standard measure of the size of the molecule. In this context, $\left< \Rog \right>$ is well known to be given in terms of the Kirchhoff index~\cite{Klein:1993tb,Klein:2002vx,bonchev_molecular_1994,Gutman:1996hq,klein_random_2004,Bapat:cx,Zhang:2006bm,babic_resistance-distance_2002} or the eigenvalues of the graph Laplacian~\cite[eq.\ 18a]{Eichinger1980}:
\begin{equation}
\left<\Rog(\graphG)\right> = \frac{d}{\verticesV^2} \Kirchhoff(\graphG) = \frac{d}{\verticesV} \tr L^+ = \frac{d}{\verticesV} \sum_{i=1}^{\verticesV - 1} \frac{1}{\lambda_i} \label{eq:kirchhoff}
\end{equation}
where $\graphG$ is a graph with $\verticesV$ vertices and $\edgesE$ edges that is embedded in $\R^d$ and the $\lambda_i$ are the nonzero eigenvalues of the graph Laplacian $L = L(\graphG)$. Especially when scaling, it's more useful to consider the contraction factor (cf.~\cite[1.48]{stepto:2015be}) of the network:
\[
g(\graphG) \ceq \frac{\left<\Rog(\graphG)\right>}{\left<\Rog(\text{path graph with $\verticesV$ vertices})\right>}.
\]
In previous work~\cite{Cantarella2022ROG}, we considered the asymptotic contraction factor for a subdivided graph whose edges obey Gaussian potentials, showing 
\begin{theorem}[\cite{Cantarella2022ROG}, Theorem 5]
\label{thm:asymptotic contraction factor}
For any connected multigraph $\graphG$ (including loop and multiple edges), if $\graphG_n$ is the $n$-part edge subdivision of $\graphG$, then the contraction factor $g(\graphG_n)$ of a Gaussian topological polymer with underlying graph $\graphG_n$ embedded in $\R^d$ obeys
\begin{equation*}
g(\graphG_\infty) \ceq \lim_{n \rightarrow \infty} g(\graphG_n) = \frac{3}{\edgesE(\graphG)^2} \left( \operatorname{Tr} \mathcal{L}^+(\graphG) + \frac{1}{3} \operatorname{Loops}(\graphG) - \frac{1}{6} \right),
\end{equation*}
where $\operatorname{Loops}(\graphG) = \edgesE(G) - \verticesV(G) + 1$ is the cycle rank of $\graphG$ (which is equal to the cycle rank of $\graphG_n$).
\end{theorem}

Here, $\mathcal{L}$ is the \emph{normalized} or~\emph{degree-weighted} graph Laplacian. 
\[
	\mathcal{L}_{ij} = \begin{cases}
					1 - \frac{2\times\text{\# loop edges}}{\deg(\vertex_i)}, & \text{ if } i = j, \\
					-\frac{k}{\sqrt{\deg(\vertex_i) \deg(\vertex_j)}}, &\text{ if } \vertex_i, \vertex_j \text{ joined by $k$ edges},\\
					0, & \text{ otherwise.}
				\end{cases}
\]

In this paper, we use a new theorem about radii of gyration to give an exact version of the formula for every $n$. The asymptotic result then follows as an easy corollary. Our main result is:

\begin{theorem}
\label{thm:main}
Suppose that embeddings $X$ of $\graphG$ are distributed as a Gaussian topological polymer with edge variance $\sigma^2$, and that $\graphG$ is a subdivision of $\graphG'$, where each edge in $\graphG'$ has been subdivided into $n$ pieces. Let $\Omega' = (\deg \vertex'_1 + \nicefrac{2}{n-1}, \dotsc, \deg \vertex'_{\verticesV'} + \nicefrac{2}{n-1})$, $\mathcal{L}_{\Omega'}(\graphG')$ be the weighted graph Laplacian of $\graphG'$ defined in \autoref{prop:weighted rog and weighted laplacian} below, and $\operatorname{Loops}(\graphG') = \edgesE' - \verticesV' + 1$ be the cycle rank of $\graphG'$ (which is equal to the cycle rank of $\graphG$).

Then the expected (squared) radius of gyration
\[
\left< \Rog(X) \right> = d \sigma^2 \frac{n^2-1}{2 \verticesV} \left( \frac{n}{n+1} \tr \mathcal{L}^+_{\Omega'}(\graphG') +  \frac{1}{3}\left(1 - \frac{1}{2\verticesV}\right) \operatorname{Loops}(\graphG') - \frac{1}{6}\left( 1 - \frac{1}{\verticesV} \right)\right). 
\]
The contraction factor of $\graphG$ is the quotient of this quantity by the corresponding expected (squared) radius of gyration for a linear polymer with the same number of monomers and the same edge variance $\sigma^2$. This is given by
\[
g(\graphG) = \frac{3(n^2-1)}{\verticesV^2-1} \left( \frac{n}{n+1} \tr \mathcal{L}^+_{\Omega'}(\graphG') +  \frac{1}{3}\left(1 - \frac{1}{2\verticesV}\right) \operatorname{Loops}(\graphG') - \frac{1}{6}\left( 1 - \frac{1}{\verticesV} \right)\right). 
\]
\end{theorem}
This implies the asymptotic result of~\autoref{thm:asymptotic contraction factor}, since $\verticesV = (n-1) \edgesE' + \verticesV' \rightarrow n \edgesE'$ as $n \rightarrow \infty$. 

These kinds of results have a long history in polymer science, though not precisely in the form above: Kurata and Fukatsu~\cite{kuratafukatsu1964} were able to compute some asymptotic contraction factors for acyclic graphs in 1964, but only considered the constant degree case and omitted correction factors of fixed size, reducing the formula above to a multiple of $\left< \Rog(X') \right>$. Eichinger and Martin~\cite{Eichinger1978:reduction} observed that the Kirchhoff matrix for a subdivided graph was highly structured and gave a method finding eigenvalues of the Kirchhoff matrix in terms of a smaller matrix which is similar to our weighted graph Laplacian. Bonchev et al.~\cite{Bonchev:2002vy} were able to compute explicitly some subdivided graph radii of gyration (again for acyclic graphs) using combinatorial methods. 

Since the expected radius of gyration of any graph is given in terms of its Kirchhoff index, as mentioned in~\eqref{eq:kirchhoff}, we can also write our main theorem purely in terms of Kirchhoff indices if we introduce the~\emph{weighted} Kirchhoff index $\Kirchhoff(\graphG,\Omega) = \omega \tr \mathcal{L}_{\Omega}^+$. As above, $\Omega = (\omega_1, \dotsc, \omega_{\verticesV})$ is a collection of positive weights $\omega_i$, $\mathcal{L}_{\Omega} = D^{-1/2} L D^{-1/2}$ where $D$ is the diagonal matrix of $\omega_i$, and $\omega = \sum_i \omega_i$. Using~\eqref{eq:total weight} and setting $\sigma^2=1$, we convert~\autoref{thm:main} into the (equivalent) 

\begin{corollary}\label{cor:main cor}
Suppose that $\graphG$ is a subdivision of the connected graph $\graphG'$ where each edge has been subdivided into $n$ parts. Then
\[
\Kirchhoff(\graphG) = \frac{n-1}{2} \left( \frac{n (n-1)}{2} \Kirchhoff \left(\graphG',\deg v_i + \frac{2}{n-1} \right) +  \frac{n+1}{3}\left(\verticesV - \frac{1}{2}\right) \operatorname{Loops}(\graphG') - \frac{n+1}{6}\left( \verticesV - 1 \right)\right).
\]
\end{corollary}
For graphs with constant vertex degree, our weighted Kirchhoff index is a multiple of the standard Kirchhoff index. The formula then simplifies quite a bit, matching results of Gao, Luo, and Liu~\cite[Theorem 3.5]{Gao:2012jk}. Yang~\cite{Yang2016} gives a formula for the Kirchhoff index of the subdivided graph with $n = 2$ in terms of three different Kirchhoff indices with various weights. Theorem~2.3 in that paper agrees with~\autoref{cor:main cor}. Similarly, Yang and Klein~\cite{Yang:2015eq} give a formula for the Kirchhoff index of the $k$-th iterated subdivision ($n = 2^k$) in terms of the same three Kirchhoff indices. Making the appropriate substitutions, their Theorem 3.2 also agrees with~\autoref{cor:main cor}. Indeed, the weighted sums of Kirchhoff indices with different weights in both~\cite{Yang2016} and~\cite{Yang:2015eq} can be factored to produce the single weighted Kirchhoff index in \autoref{cor:main cor}. We note that~Carmona, Mitjana, and Mons\'o~\cite{Carmona:2017ge} also give a formula for the Kirchhoff index of a subdivision, but their Corollary 4.4 is stated in terms of entries in $L^+$ and so is harder to compare with the version above. 

\section{Notation and review of theory}

Throughout this paper, we will make frequent use of the notation $1_{m \times n}$ to denote the $m \times n$ matrix of all $1$s. We will use $1_m$ to denote the $m$-dimensional column vector of all $1$s.\footnote{In what follows we occasionally use the notation $1_{m \times 1}$ when the object being described is conceptually an $m \times 1$ matrix rather than a vector, even though there is no visual distinction between an $m \times 1$ matrix of $1$s and an $m$-dimensional column vector of $1$s.}

We will deal with connected, directed graphs $\graphG$ with $\verticesV$ vertices $\vertex_1, \dotsc, \vertex_{\verticesV}$ and $\edgesE$ edges $\edge_1, \dotsc, \edge_{\edgesE}$, allowing loop and multiple edges. We regard the graph as a simplicial complex\footnote{If we have loop edges or multiple edges, this is technically a $\Delta$-complex, not a simplicial complex (see~\cite[Section 2.1]{HatcherAT}). The distinction won't matter for our purposes.} and introduce real vector spaces $C_0 \simeq \R^{\verticesV}$, $C_1 \simeq \R^{\edgesE}$ whose entries are (formal) linear combinations of vertices or edges. We will also need the spaces $C^0 = \Hom(C_0, \R^d) \simeq \Mat_{d \times \verticesV}$ and $C^1 = \Hom(C_1, \R^d) \simeq \Mat_{d \times \edgesE}$ of linear maps from $C_0$ and $C_1$ to $\R^d$ (which can be written as $d \times \verticesV$ and $d \times \edgesE$ real matrices with respect to the bases we've chosen for $C_0$ and $C_1$ by numbering vertices and edges.) We'll refer to elements of $C^0$ and $C^1$ (meaning matrices or linear maps) by $X$ and $W$, respectively. 

Recall that for any $m$ and $n$, the isomorphism $\Vec \from \Mat_{m \times n} \to \R^{m \cdot n}$ is given by stacking columns atop one another, letting $\Vec(C) = (c_{11}, \dotsc, c_{m1}, c_{12}, \dotsc, c_{mn})^T$. If $A \in \Mat_{p \times n}$ and $B \in \Mat_{q \times m}$, then we have the ``vec trick'' $(A \otimes B) \Vec{C} = \Vec(B C A^T)$ where $\otimes$ is the Kronecker product of matrices (or tensor product of linear maps). When a subspace $S$ of $\Mat_{m \times n}$ is isomorphic (by $\Vec$) to a subspace $T$ of $\R^{m \cdot n}$, we'll write $S \cong T$. When subspaces $S$ and $T$ of the same ambient space ($\Mat_{m \times n}$ or $\R^{m \cdot n}$) are the same, we'll write $S=T$. In general, $X$ and $W$ will denote matrices in $\Mat_{d \times \verticesV}$ and $\Mat_{d \times \edgesE}$ while $\Vec X$ and $\Vec W$ denote the corresponding vectors in $\R^{d \cdot \verticesV}$ and $\R^{d \cdot \edgesE}$. However, we'll abuse notation slightly and use $C^0$ to denote either $\Mat_{d \times \verticesV}$ or $\R^{d \cdot \verticesV}$ (trusting the reader to figure out which one is meant from context) and do the same for $C^1$. 

Since $\graphG$ is directed, each edge $i$ travels from some vertex $\tail(i)$ to another vertex $\head(i)$ (which may be the same if $\graphG$ has loop edges).  These define a linear map called the boundary map $\bdy : C_1 \rightarrow C_0$, which has $\bdy(\edge_i) = \vertex_{\head(i)} - \vertex_{\tail(i)}$. The boundary map is given by the $\verticesV \times \edgesE$ incidence matrix of the graph $\graphG$. Here are some useful facts about $\bdy$:

\begin{proposition}[{cf.~\cite[Theorem VIII.3.4]{MasseyBasicCourse}}] \label{prop:bdy facts}
The space $\ker \bdy$ is the linear subspace spanned by all cycles in $\graphG$. The dimension of this vector space is variously called the first Betti number, cycle rank, or first homology of $\graphG$. It is given by 
\[
\Loops(\graphG) = \chi(\graphG) = \edgesE - \verticesV + 1.
\]
Since $\graphG$ is connected, the space $\ker \bdy^T$ is the one-dimensional space spanned by $1_{\verticesV}$.
\end{proposition}

\begin{definition}
A collection of~\emph{vertex positions} for $\graphG$ in $\R^d$ is given by $X \in C^0$, where the position of the $i$-th vertex $x_i = X(\vertex_i)$ is the $i$-th column of the matrix $X$. A collection of~\emph{edge displacements} for $\graphG$ in $\R^d$ is given by $W \in C^1$, where the displacement along the $i$-th edge $w_i = W(\edge_i)$ is the $i$-th column of the matrix $W$. We say that $W$ and $X$ are~\emph{compatible} if $w_i = x_{\head(i)} - x_{\tail(i)}$. 
\end{definition}
Unwinding definitions, it now follows immediately that
\begin{equation}
W \text{ and } X \text{ are compatible} \iff W = X \bdy \iff \Vec{W} = (\bdy^T \otimes I_d) \Vec{X}. \label{eq:boundary and kronecker}
\end{equation}
We first analyze the kernel and image of the matrix $\bdy^T \otimes I_d$. These are subspaces of $\R^{d \cdot \edgesE}$, but it is easier to describe them as subspaces of $\Mat_{d \times \edgesE}$ using tensor products.\footnote{Given a subspace $U \subset \R^n$, these match the usual definitions of the tensor product $U^* \otimes \R^m$ as a linear subspace of $(\R^n)^* \otimes \R^m \simeq \Hom(\R^n,\R^m) \simeq \Mat_{m \times n}$ (see~\cite[Theorem 14.6]{RomanAdvancedLinearAlgebra}), where $^*$ denotes the dual space. We have dropped the $^*$ below to avoid confusion.}
\begin{lemma} \label{lem:ker and image as tensors}
The kernel and image of $\bdy^T \otimes I_d$ obey
\[
\begin{aligned}
\ker (\bdy^T \otimes I_d) &\cong (\ker \bdy^T) \otimes \R^d = \{ X \in \Mat_{d \times \verticesV} \mid \text{ each column vector of $X^T \in \ker \bdy^T$} \} \\
\im (\bdy^T \otimes I_d) &\cong (\im \bdy^T) \otimes \R^d = \{ W \in \Mat_{d \times \edgesE} \mid \text{ each column vector of $W^T \in \im \bdy^T$} \},
\end{aligned}
\]
where we remind the reader that $\cong$ means ``isomorphic by $\Vec$''.
\end{lemma}

\begin{proof} We know $\Vec W = (\bdy^T \otimes I_d) \Vec{X} \iff W = X \bdy$. Equivalently, $W^T = \bdy^T X^T$.
\end{proof}

This lemma implies that not every collection of edge displacements $W$ is compatible with a collection of vertex positions $X$. This is not surprising: if $\graphG$ has cycles, then the sum of edge displacements around each cycle must vanish. 

\begin{proposition} \label{prop:ker and image described}
We have 
\[
\begin{aligned}
\im (\bdy^T \otimes I_d) &\cong (\ker \bdy)^\perp \otimes \R^d \\
\ker (\bdy^T \otimes I_d) &\cong (\ker \bdy^T) \otimes \R^d = \{ 1_{1 \times \verticesV} \otimes u \in \Mat_{d \times \verticesV} \mid u \in \R^d \},
\end{aligned}
\]
while $\dim (\im \bdy^T \otimes I_d) = d (\edgesE - \Loops(\graphG))$ and $\dim (\ker \bdy^T \otimes I_d) = d$.
\end{proposition}

\begin{proof}
By~\autoref{lem:ker and image as tensors}, $\im (\bdy^T \otimes I_d) \cong (\im \bdy^T) \otimes \R^d$. It's standard that $(\im \bdy^T) \otimes \R^d = (\ker \bdy)^\perp \otimes \R^d$. By~\autoref{prop:bdy facts}, $(\ker \bdy)^\perp$ has dimension $\edgesE - \Loops(\graphG)$. Tensoring with $\R^d$ multiplies this dimension by $d$, as each of the $d$ rows of any $W \in (\ker \bdy)^\perp \otimes \R^d$ lies in $(\ker \bdy)^\perp$.  

For the second part, $\ker (\bdy^T \otimes I_d) = (\ker \bdy^T) \otimes \R^d$ by~\autoref{lem:ker and image as tensors}. Suppose $X$ is any matrix in $(\ker \bdy^T) \otimes \R^d$. For each $i \in 1, \dotsc, d$, the $i$-th row of  $X$ lies in $\ker \bdy^T$ and by~\autoref{prop:bdy facts} is given by $u_i 1_{\verticesV}$ for some $u_i$. This means that $X = 1_{1 \times \verticesV} \otimes u$ where $u$ is the column vector of $u_i$, as claimed. \end{proof}

We have seen how to construct a compatible $W \in C^1$ from $X \in C^0$. To go in the other direction, we may use the Moore–Penrose pseudoinverse, writing $\Vec X = (\bdy^T \otimes I_d)^+ \Vec W$. Therefore, we now compute the kernel and image of $(\bdy^T \otimes I_d)^+$. 

\begin{lemma} \label{lem:pseudoinverse kernel and image} We have
\[
\begin{aligned}
\im (\bdy^{T} \otimes I_d)^+ &\cong \{X \in C^0 \mid X 1_{\verticesV} = 0 \} = \{ X \in C^0 \mid \text{ the center of mass } \frac{1}{\verticesV} \sum_i x_i = 0 \}\\
\ker (\bdy^T \otimes I_d)^+ &\cong \ker \bdy \otimes \R^d = \{ W \in \Mat_{d \times \edgesE} \mid \text{ each column vector of $W^T \in \ker \bdy$} \},
\end{aligned} 
\]
where $\dim \im (\bdy^T \otimes I_d)^+ = d(\verticesV - 1)$ and $\dim \ker (\bdy^T \otimes I_d)^+ = d \cdot \Loops(\graphG)$.
\end{lemma}

\begin{proof} We first compute (the $\cong$ is analogous to the proof of~\autoref{lem:ker and image as tensors}) 
\[
\im  (\bdy^{T} \otimes I_d)^+ = \im(\bdy^{T} \otimes I_d)^T = \im (\bdy \otimes I_d) \cong (\im \bdy) \otimes \R^d = (\ker \bdy^T)^\perp \otimes \R^d.
\]
Since $\ker \bdy^T = \Span 1_{\verticesV}$ by~\autoref{prop:bdy facts}, this means that $X \in (\ker \bdy^\perp) \otimes \R^d$ if and only if each row of $X$ is orthogonal to $1_{1 \times \verticesV}$. But this is exactly the condition $X 1_{\verticesV} = 0$. Of course, $X 1_{\verticesV}$ is also the sum of the columns of $X$, so it is equal to $\sum_{i=1}^{\verticesV} x_i$.

On the other hand, $\ker (\bdy^T \otimes I_d)^+ = \ker (\bdy^T \otimes I_d)^T = \ker (\bdy \otimes I_d) \cong (\ker \bdy) \otimes \R^d$. We know $\dim (\ker \bdy)^\perp = \verticesV - \dim \ker \bdy^T = \verticesV - 1$ and $\dim \ker \bdy = \Loops(\graphG)$ from~\autoref{prop:bdy facts}, and tensoring with $\R^d$ multiplies each dimension by $d$.
\end{proof}

We are now ready to give the construction of a Gaussian topological polymer in our terms; it is a probability distribution on $\im (\bdy^T \otimes I_d)^+$ given by pushing forward a Gaussian distribution on $C^1$ conditioned on the hypothesis that $W \in \im (\bdy^T \otimes I_d) \subsetneq C^1$. Because Gaussians are very special, the conditional distribution on $\im (\bdy^T \otimes I_d)$ is the same as the pushforward distribution under the orthogonal projection to $\im (\bdy^T \otimes I_d)$. 

\begin{definition}\label{def:projector P}
Let $P : C^1 \to C^1$ be the orthogonal projector to $\im (\bdy^T \otimes I_d) \subset C^1$. We note that 
\[
P = (\bdy^T \otimes I_d)(\bdy^T \otimes I_d)^+ = (\bdy^T \otimes I_d)(\bdy^{T+} \otimes I_d) = 
(\bdy^T \bdy^{T+}) \otimes I_d.
\]
\end{definition}

\begin{definition}
\label{defn:phantom network distribution}
If $\Vec W \dist \mathcal{N}(0,\sigma^2 P)$ and $\Vec X = (\bdy^T \otimes I_d)^+ \Vec W$, we say $X$ is distributed as a~\emph{Gaussian topological polymer} with edge variance $\sigma^2$. We note that the covariance matrix of $\Vec X$ is given explicitly by
\[
\begin{aligned}
(\bdy^T \otimes I_d)^+ (\sigma^2 P) (\bdy^T \otimes I_d)^{T+} &= \sigma^2 (\bdy^{T+} \otimes I_d) (\bdy^T \bdy^{T+} \otimes I_d)(\bdy^+ \otimes I_d) \\
&= \sigma^2 (\bdy^{T+} \bdy^T \bdy^{T+} \bdy^+) \otimes I_d = \sigma^2 \bdy^{T+} \bdy^+ \otimes I_d = \sigma^2 (\bdy \bdy^T)^+ \otimes I_d \\
&= \sigma^2 L^+ \otimes I_d.
\end{aligned}
\]
where $L = \bdy \bdy^T$ is variously known as the graph Laplacian, the combinatorial Laplacian, or the Kirchhoff matrix of the graph $\graphG$, so we may say $X \dist \mathcal{N}(0,\sigma^2 L^+ \otimes I_d)$.
\end{definition}
We point out that the covariance matrix $\sigma^2 P$ for $\Vec W$ is usually~\emph{not} positive definite because $\Vec W$ is supported on the subspace $\im (\bdy^T \otimes I_d)$ of $C^1$. It's more standard in the physics literature to write the probability density function for $X$ instead the covariance matrix. In this form the density of $X$ is proportional to
\[
e^{-\frac{1}{2} (\Vec X)^T (\sigma^2 L \otimes I_d) (\Vec X)} = e^{-\frac{\sigma^2}{2} \tr(XLX^T)} = e^{-\sigma^2 \sum_{i = 1}^{\edgesE} \norm{x_{\head(i)} - x_{\tail(i)}}^2}
\]
but we again note that this form can be misleading. $L$ is positive semi-definite (but not positive definite) because the density is supported \emph{only} on the subspace of $C^0$ given by the centered distributions $\Vec X \in \im(\bdy^T \otimes I_d)^+ \cong \{ X \mid X 1_{\verticesV} = 0\}$. This is the ensemble of embeddings classically considered in~(\cite{James:1947hp}, \cite{Flory1976}, equations (14) and (15)).
In this ensemble, we treat the edges of $\graphG$ as Gaussian springs.  

\section{Gaussian topological polymers based on subdivision graphs}

Now suppose that $\graphG'$ is a (multi-)graph with $\verticesV'$ vertices and $\edgesE'$ edges, and that $\graphG$ is a $\verticesV$ vertex and $\edgesE$ edge graph created by subdividing each edge of $\graphG'$ into $n$ parts. It is clear that $\edgesE = n \edgesE'$. We can divide the vertices of $\graphG$ into $\verticesV'$ ``junction'' vertices, which correspond to vertices in $\graphG'$, and $\verticesV - \verticesV'$ ``interior'' vertices created by the subdivision. Since the division of each edge of $\graphG'$ into $n$ parts creates $n-1$ new interior vertices in $\graphG$, we have $\verticesV = \verticesV' + (n-1) \, \edgesE'$. (It follows immediately that $\Loops(\graphG') = \Loops(\graphG)$.) 
Further, there are linear maps $f_0 : C_0(\graphG') \to C_0(\graphG)$, which takes vertices in $\graphG'$ to the corresponding vertices in $\graphG$, and $f_1: C_1(\graphG') \to C_1(\graphG)$, which takes each edge $\edge'_i \in \graphG'$ to the (formal) sum of the $n$ edges in $\graphG$ created by subdividing $\edge'_i$. These relationships are shown in~\autoref{fig:structure}.

Any embedding $X$ of $\graphG$ induces a corresponding embedding $\Vec X' = (f_0^T \otimes I_d) \Vec X$; this just assigns to each $\vertex'_i$ in $\graphG'$ the position of the corresponding junction vertex $f_0(\vertex'_i)$ in $\graphG$. Suppose that $X$ is distributed as a Gaussian topological polymer. What is the distribution on $X'$? It would be quite natural to hope that $X'$ is also distributed as a Gaussian topological polymer-- where the edge displacements of $X$ are chosen from Gaussians (conditioned on loop closures), the edge displacements of $X'$ are chosen from sums of Gaussians (conditioned on the closure of corresponding loops). 

This is not quite the case: if $X$ is distributed as a Gaussian topological polymer, it always has center of mass $\sum_{i=1}^{\verticesV} x_i$ at the origin, but there is no reason to suppose that the center of mass of the subset of $x_i$ corresponding to junction vertices is~\emph{also} at the origin. To fix this, we introduce the ``zero-mean'' matrix $Z_m = I_{m} - \frac{1}{m} 1_{m \times m}$ which orthogonally projects each $u \in \R^m$ to the closest vector with $\sum u_i = 0$. 

\begin{figure}[t]
\hfill
\begin{overpic}[width=0.4\textwidth]{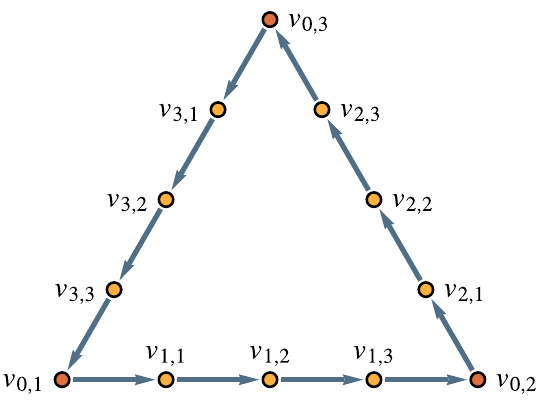} 
\end{overpic}
\hfill
\begin{overpic}[width=0.4\textwidth]{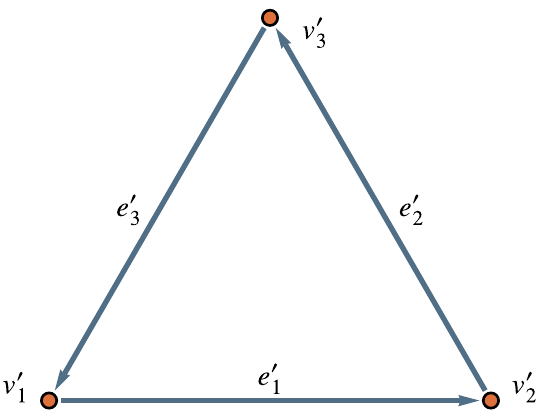} 
\end{overpic}
\hfill
\hphantom{.}
\caption{The graph $\graphG$ (at left) is a subdivision of the graph $\graphG'$ at right. The vertices $\vertex_{i,j}$ and edges $\edge_{i,j}$ of $\graphG$ are numbered to correspond with vertices and edges of $\graphG'$; each $\vertex_{0,j}$ corresponds to a vertex $\vertex'_j$ of $\graphG'$, while vertices $\vertex_{i,1}, \dotsc, \vertex_{i,n-1}$ are those created by subdividing edge $\edge'_i$ of $\graphG'$ into $n$ new edges. The edges $\edge_{i,j}$ of $\graphG$ aren't labeled in the picture, but are constructed so that $\edge_{i,1}, \dotsc, \edge_{i,n}$ are the edges created by subdividing $\edge'_i$. Here $n=4$, $\verticesV = 22$, $\edgesE = 24$, $\verticesV' = 4$ and $\edgesE'= 6$ so we can check that $\edgesE = n \edgesE'$ and $\verticesV = \verticesV' + (n-1) \edgesE'$. The map $f_0(\vertex'_i) = \vertex_{0,i}$ while $f_1(\edge'_i) = \sum_{j=1}^{n} \edge_{i,j}$. }
\label{fig:structure}
\end{figure}

\begin{proposition} \label{prop:embeddings of structure graph}
If $X$ is distributed as a Gaussian topological polymer (for $\graphG$) with edge variance $\sigma^2$ and $\Vec X' = (f_0^T \otimes I_d) \Vec X$, then $(Z_{\verticesV'} \otimes I_d) \Vec X'$ is distributed as a Gaussian topological polymer (for $\graphG'$) with edge variance $n \sigma^2$. 
\end{proposition}

The rest of this section is devoted to giving a formal proof of the (intuitively obvious) proposition above and involves some technical detail; readers willing to accept the statement may safely skip ahead.  

\begin{lemma} \label{lem:chain map properties}
If $\bdy$ and $\bdy'$ are the boundary maps for $\graphG$ and $\graphG'$, $f_1^T f_1 = n I_{\edgesE'}$ and $\bdy f_1 = f_0 \bdy'$.
\end{lemma}

\begin{proof}
We may assume that the edges of $\graphG$ and $\graphG'$ are numbered and oriented so that the $n \cdot \edgesE' \times \edgesE'$ matrix $f_1 = I_{\edgesE'} \otimes 1_{n \times 1}$. Then it's clear that $f_1^T f_1 = n I_{\edgesE'}$. (Renumbering and/or reorienting the edges of $\graphG$ and $\graphG'$ would conjugate $f_1$ by orthogonal matrices and leave $f_1^T f_1$ unchanged.)

Each column of $f_1$ represents a single subdivided edge: applying $\bdy$ to each column of $f_1$ we get the difference of the vertices at the ends of the subdivided edge. These junction vertices are the endpoints of the corresponding  (undivided) edge in $\graphG'$; that is, the output when applying $f_0$ to $\delta'$. 
\end{proof}

\begin{lemma}
\label{lem:homotopy equivalent}
Suppose $P = \bdy^{T} \bdy^{T+} \otimes I_d$ and $P' = (\bdy')^T (\bdy')^{T+} \otimes I_d$. Then $(f_1^T \otimes I_d) P = P' (f_1^T \otimes I_d)$.
\end{lemma}

\begin{proof}
We start by showing that $f_1 (\ker \bdy') = \ker \bdy$. If $w \in \ker \bdy'$, we have $\bdy f_1 w = f_0 \bdy' w = 0$. Therefore, $f_1 (\ker \bdy') \subset \ker \bdy$. 

On the other hand, $f_1 = I_{\edgesE'} \otimes 1_{n \times 1}$, so it is clear that $\operatorname{rank} f_1 = \edgesE'$. It follows that $\dim \ker f_1 = 0$, and that $\dim f_1(\ker \bdy') = \dim \ker \bdy'$.  However, $\dim \ker \bdy' = \Loops(\graphG') = \Loops(\graphG) = \dim \ker \bdy$. Therefore $f_1 (\ker \bdy') = \ker \bdy$. 

Notice that, since $I - (\bdy')^+ (\bdy')$ is the orthogonal projector onto $\ker \bdy'$, we have shown $f_1 (\ker \bdy') = \im f_1 (I - (\bdy')^+(\bdy')) = \ker \bdy$. Since $I - \bdy^+ \bdy$ is the orthogonal projector onto $\ker \bdy$ it follows that
\begin{equation}
f_1 (I - (\bdy')^+(\bdy')) = (I - \bdy^+ \bdy) f_1 (I - (\bdy')^+(\bdy')). \label{eq:part1}
\end{equation}

We now claim that $f_1((\ker \bdy')^\perp) \subset (\ker \bdy)^\perp$. Suppose that $w_1' \in (\ker \bdy')^\perp$, and $w_2 \in \ker \bdy$. Since $f_1 (\ker \bdy') = \ker \bdy$, we know $w_2 = f_1 w_2'$ for some $w_2' \in \ker \bdy$. We compute
\[
\langle f_1 w_1', w_2 \rangle = \langle f_1 w_1', f_1 w_2' \rangle = \langle w_1', (f_1^T f_1) w_2' \rangle = \langle w_1', n w_2' \rangle = n \langle w_1', w_2' \rangle = 0,
\]
where we used~\autoref{lem:chain map properties} in the third equality. Since $w_2 \in \ker \bdy$ was arbitrary, this proves that $f_1 w_1' \in (\ker \bdy)^\perp$, so we have proved that $f_1((\ker \bdy')^\perp) \subset (\ker \bdy)^\perp$. 

Now $(\bdy')^+(\bdy')$ is the orthogonal projector onto $(\ker \bdy')^\perp$, so $f_1((\ker \bdy')^\perp) = \im f_1 (\bdy')^+ (\bdy') \subset (\ker \bdy)^\perp$. Since $\bdy^+ \bdy$ is the orthogonal projector onto $(\ker \bdy)^\perp$, it follows that 
\begin{equation} 
f_1 (\bdy')^+ (\bdy') = \bdy^+ \bdy f_1 (\bdy')^+ (\bdy'). \label{eq:part2}
\end{equation}
Expanding the right hand side of~\eqref{eq:part1} and using~\eqref{eq:part2}, we have
\[
\begin{aligned}
f_1 (I - (\bdy')^+(\bdy')) &= (I - \bdy^+ \bdy) f_1 (I - (\bdy')^+(\bdy')) \\
 &= 
 f_1 - \bdy^+ \bdy f_1 - f_1 (\bdy')^+ (\bdy') + \bdy^+ \bdy f_1 (\bdy')^+ (\bdy') \\
 &= f_1 - \bdy^+ \bdy f_1. 
\end{aligned}
\]
It follows that $f_1 (\bdy')^+ (\bdy') = \bdy^+ \bdy f_1$ or equivalently that $(\bdy')^T (\bdy')^{T+} f_1^T = f_1^T \bdy^T \bdy^{T+}$. But this completes the proof, because now
\begin{multline*}
(f_1^T \otimes I_d) P = (f_1^T \otimes I_d)(\bdy^{T} \bdy^{T+} \otimes I_d) 
= f_1^T \bdy^{T} \bdy^{T+} \otimes I_d = (\bdy')^T (\bdy')^{T+} f_1^T \otimes I_d \\
= ((\bdy')^T (\bdy')^{T+} \otimes I_d)(f_1^T \otimes I_d) = P' (f_1^T \otimes I_d),
\end{multline*}
as required.
\end{proof}

We're now ready to prove~\autoref{prop:embeddings of structure graph}.
\begin{proof}[{Proof of \autoref{prop:embeddings of structure graph}}]
From~\autoref{prop:bdy facts}, we know that $\ker (\bdy')^T$ is spanned by $1_{\verticesV'}$. Orthogonal projection onto this space is therefore given by $\frac{1_{\verticesV' \times \verticesV'}}{\verticesV'}$ and so the zero-mean matrix $Z_{\verticesV'}$ is orthogonal projection onto $(\ker (\bdy')^T)^\perp = \im \bdy'$. It follows that
\[
Z_{\verticesV'} = (\bdy') (\bdy')^{+} = (\bdy')^{T+} (\bdy')^T \implies Z_{\verticesV'} \otimes I_d = ((\bdy')^T \otimes I_d)^+ ((\bdy')^T \otimes I_d).
\] 
Further, from~\autoref{lem:chain map properties}, we have 
\[
\bdy f_1 = f_0 \bdy' \implies f_1^T \bdy^T = (\bdy')^T f_0^T \implies (f_1^T \otimes I_d)(\bdy^T \otimes I_d) = ((\bdy')^T \otimes I_d)(f_0^T \otimes I_d).
\]
We are now ready to compute. For convenience, let $\Vec W' = (f_1^T \otimes I_d) P \Vec W$. We have
\[
\begin{aligned}
(Z_{\verticesV'} \otimes I_d) \Vec X' &= (Z_{\verticesV'} \otimes I_d) (f_0^T \otimes I_d) \Vec X \\
&=  ((\bdy')^T \otimes I_d)^+ ((\bdy')^T \otimes I_d) (f_0^T \otimes I_d) \Vec X \\
&= ((\bdy')^T \otimes I_d)^+  (f_1^T \otimes I_d)(\bdy^T \otimes I_d) \Vec X \\
&= ((\bdy')^T \otimes I_d)^+  (f_1^T \otimes I_d) (\bdy^T \otimes I_d)(\bdy^T \otimes I_d)^+ \Vec W \\
&= ((\bdy')^T \otimes I_d)^+  (f_1^T \otimes I_d) P \Vec W \\
&= ((\bdy')^T \otimes I_d)^+ \Vec W'.
\end{aligned}
\]
We have proved that $(Z_{\verticesV'} \otimes I_d) \Vec X'$ is the pushforward of $\Vec W'$ under the map $((\bdy')^T \otimes I_d)^+$. We now need to investigate the distribution of $W'$. We know that $\Vec W \dist \mathcal{N}(0,\sigma^2 P)$. Thus $\Vec W'$ is also normal, and also has mean zero. The covariance matrix of $\Vec W'$ is given by 
\[
\begin{aligned}
(f_1^T \otimes I_d) P (\sigma^2 P) ((f_1^T \otimes I_d) P)^T = \sigma^2 (f_1^T \otimes I_d) P (f_1 \otimes I_d)
= \sigma^2 P' (f_1^T f_1 \otimes I_d) = n \sigma^2 P'. 
\end{aligned}
\]
where we used the fact that $P^T = P$ and $P^2 = P$ in the first equality,~\autoref{lem:homotopy equivalent} in the second equality and~\autoref{lem:chain map properties} in the third. Thus $\Vec W' \dist \mathcal{N}(0,n \sigma^2 P')$, as required.
 
\end{proof}

Suppose that we associate to each $\edge'_i \in \graphG'$ the corresponding set $\mathcal{E}_i$ of $n$ edges in $\graphG$ created by subdividing $\edge'_i$, and let $\mathcal{S}$ be the group of $\edgesE \times \edgesE$ permutation matrices which preserve all of the $\mathcal{E}_i$. Any $S \in \mathcal{S}$ defines a map $S \otimes I_d$ from $C^1(\graphG)$ to itself, and a corresponding map from $C^0(\graphG)$ to itself; we denote this by $X \mapsto X^S$, where
\[
X^S = (\bdy^{T+} \otimes I_d) (S \otimes I_d) (\bdy^T \otimes I_d) X.
\]
\begin{lemma}
\label{lem:invariance}
If $X \dist \mathcal{N}(0,\sigma^2 L^+ \otimes I_d)$ with any edge variance $\sigma^2$ and $S \in \mathcal{S}$, then 
\[
X^S \dist  \mathcal{N}(0,\sigma^2 L^+ \otimes I_d).
\]
\end{lemma}

\begin{proof}
Fix any $S \in \mathcal{S}$. If $X \dist \mathcal{N}(0,\sigma^2 L^+ \otimes I_d)$, we know that $(\bdy^{T+} \otimes I_d) (S \otimes I_d) (\bdy^T \otimes I_d) X$ is also normally distributed with mean $0$ and covariance matrix given by the matrix product
\begin{equation}
(\bdy^{T+} \otimes I_d) (S \otimes I_d) (\bdy^T \otimes I_d)(\sigma^2 L^+ \otimes I_d) ((\bdy^{T+} \otimes I_d) (S \otimes I_d) (\bdy^T \otimes I_d))^T. \label{eq:covariance} 
\end{equation}
We start by considering the innermost terms
\[
\begin{aligned}
(\bdy^T \otimes I_d)(\sigma^2 L^+ \otimes I_d)(\bdy \otimes I_d) &= 
\sigma^2 (\bdy^T L^+ \bdy) \otimes I_d \\
&= \sigma^2 (\bdy^T \bdy^{T+} \bdy^+ \bdy) \otimes I_d
\end{aligned}
\]
Now $\bdy^+ \bdy = (\bdy^+ \bdy)^T = \bdy^T \bdy^{T+}$ by the Moore-Penrose properties. So we may rewrite this as 
\[
\begin{aligned}
\sigma^2 (\bdy^T \bdy^{T+} \bdy^+ \bdy) \otimes I_d &= \sigma^2 (\bdy^T \bdy^{T+} \otimes I_d)(\bdy^T \bdy^{T+} \otimes I_d) \\
&= \sigma^2PP = \sigma^2 P
\end{aligned}
\]
We proved in~\cite{CantarellaSchumacherShonkwiler2024b} that $S \otimes I_d$ fixes $\im (\bdy^T \otimes I_d)$. Since $S \otimes I_d$ is a permutation matrix, it is an orthogonal matrix: this means it also fixes $\ker (\bdy \otimes I_d)^\perp$. In turn, this means that $S \otimes I_d$ commutes with $P$, which is orthogonal projection to $\im (\bdy^T \otimes I_d)$. So we may return to~\eqref{eq:covariance} and observe that it's now equal to 
\[
\begin{aligned}
(\bdy^{T+} S (\sigma^2 P) S^T \bdy^{+}) \otimes I_d &= \sigma^2 (\bdy^{T+} P S S^T \bdy^+) \otimes I_d \\
&= \sigma^2 (\bdy^{T+} P \bdy^+) \otimes I_d \\
&= \sigma^2 (\bdy^{T+} \bdy^T \bdy^{T+} \bdy^+) \otimes I_d \\
&= \sigma^2 (\bdy^{T+} \bdy^+) \otimes I_d = \sigma^2 L^+ \otimes I_d,
\end{aligned}
\]
as desired, using the Moore--Penrose properties again to go from the third to fourth lines. 
\end{proof}

\section{Weighted Radius of Gyration and Weighted Graph Laplacians}

We start by recalling the definition of the weighted radius of gyration:

\begin{definition}
\label{defn:radius of gyration}
If $X = (x_1, \dotsc, x_{\verticesV})$ is a $d \times \verticesV$ matrix whose columns are points in $\R^d$ and $\Omega = (\omega_1, \dotsc, \omega_\verticesV) \in \R^{\verticesV}_+$ is a corresponding collection of positive scalars, we define the~\emph{total weight} $\abs{\varOmega} = \sum \omega_i$, the~\emph{center of mass} or~\emph{expectation} $\mu(X,\varOmega) = \frac{1}{\abs{\varOmega}} \sum_i \omega_i \, x_i$, and the~\emph{weighted radius of gyration} or~\emph{variance} to be $\Rog(X,\varOmega)$, given by
\[
\Rog(X,\varOmega)\ceq\frac{1}{2\abs{\varOmega}^2} \sum_{i=1}^{\verticesV} \sum_{j=1}^{\verticesV} \omega_i \, \omega_j \norm{x_i - x_j}^2 .
\]
\end{definition}

It is well-known that the expectation of the (unweighted) radius of gyration of a graph embedding distributed as a Gaussian topological polymer is given in terms of the trace of the pseudoinverse of the (unweighted) graph Laplacian~\cite[eq.\,18a]{Eichinger1980}, which is also known as the Kirchhoff index. The statement has been generalized to the degree-weighted radius of gyration and degree-weighted (or normalized) graph Laplacian~\cite[p.\,1694]{Chen:2010da} as well. In fact, the analogous statement holds for arbitrary weights, as we now intend to show. 

It is a classical observation that the unweighted radius of gyration is a quadratic form evaluated on vertex positions~(cf.~\cite{Zimm1949}). We now derive the corresponding formula for~\emph{weighted} radius of gyration.

\begin{lemma}
\label{lem:radius of gyration form}
If $X = (x_1, \dotsc, x_{\verticesV})$ is a $d \times \verticesV$ matrix whose columns are points in $\R^d$ and $\Omega = (\omega_1, \dotsc, \omega_\verticesV) \in \R^{\verticesV}_+$ is a corresponding collection of positive scalars, let $D = \diag(\omega_1, \dotsc, \omega_{\verticesV})$. Then the weighted radius of gyration is the quadratic form 
\[
\Rog(X,\Omega) = (\operatorname{vec} X)^T (Q \otimes I_d) (\operatorname{vec} X), \text{ where } Q = \frac{1}{\tr D} \left( D- \frac{(D 1_{\verticesV})(D 1_{\verticesV})^T}{\tr D} \right).
\]
\end{lemma}

\begin{proof}
This follows mechanically from expanding the formula in~\autoref{defn:radius of gyration}:
\[
\begin{aligned}
\sum_{i=1}^n \sum_{j=1}^n\omega_i \, \omega_j \norm{x_i - x_j}^2 &= \sum_{i=1}^{\verticesV} \sum_{j=1}^{\verticesV} \sum_{k=1}^{d} \omega_i \, \omega_j (x_{ki}^2 + x_{kj}^2 - 2 x_{ki} x_{kj}) \\
&= 2 \left(\sum_{i=1}^{\verticesV} \omega_i\right) \sum_{j=1}^{\verticesV} \sum_{k=1}^{d} \omega_j x_{kj}^2 - 
2 \sum_{k=1}^d \sum_{i=1}^{\verticesV} \sum_{j=1}^{\verticesV} \omega_i \omega_j x_{ki} x_{kj}.
\end{aligned}
\]
This expression is quadratic in the entries of $X$, so it can be written as $(\operatorname{vec} X)^T C (\operatorname{vec} X)$ for~\emph{some} $d \verticesV \times d \verticesV$ matrix of coefficients $C$. To complete the proof, we compute $C$. The first sum contributes the diagonal matrix $D \otimes I_d$ to $C$. Observing that the  
$\verticesV \times \verticesV$ matrix whose $ij$-th entry is $\omega_i \omega_j$ is given explicitly by the rank-1 matrix $(D 1_{\verticesV}) (D 1_\verticesV)^T$, we see the second sum contributes $(D 1_{\verticesV}) (D 1_\verticesV)^T \otimes I_d$ to $C$. Using $\abs{\Omega} = \sum_{i=1}^{\verticesV} \omega_i = \tr D$ and dividing through by $2 (\tr D)^2$ completes the proof.
\end{proof}

We can now give our generalization of the relationship between radius of gyration and Kirchhoff index:

\begin{proposition}
If $X$ is distributed as a Gaussian topological polymer (for $\graphG$) in $\R^d$ with edge variance $\sigma^2$, then the expected value of the weighted radius of gyration $\langle \Rog(X,\Omega) \rangle$ for any set of weights $\omega_i>0$ on the vertices of $\graphG$ is given by
\[
\langle \Rog(X,\Omega) \rangle = \frac{d \sigma^2}{\omega} \tr \mathcal{L}_\Omega^+(\graphG)
\]
where $\mathcal{L}_{\Omega}(\graphG)$ is given in terms of the ordinary graph Laplacian $L(\graphG)$ and the diagonal matrix of weights $D = \operatorname{diag}(\omega_1, \dots, \omega_{\verticesV})$ by $\mathcal{L}_{\Omega}(\graphG) = D^{-1/2} L(\graphG) D^{-1/2}$, and $\omega$ is the total weight $\omega = \sum_i \omega_i$.
\label{prop:weighted rog and weighted laplacian}
\end{proposition}

\begin{proof}
From~\autoref{defn:phantom network distribution}, we know $\Vec X$ has mean $0$ and covariance matrix $L^+ \otimes I_d$. Since~$\Rog(X,\Omega)$ is a quadratic form in $X$ with coefficient matrix $Q \otimes I_d$ it follows that the expectation 
\[
\langle \Rog(X,\Omega) \rangle = \tr((Q \otimes I_d)(\sigma^2 L^+ \otimes I_d)) = d\sigma^2 \tr (Q L^+).
\]
Now the matrix $Q$ can be rewritten
\[
Q = \frac{1}{\tr D} \left( D- \frac{(D 1_{\verticesV})(D 1_{\verticesV})^T}{\tr D} \right) = 
\frac{1}{\tr D} D^{1/2} \left( I - \frac{(D^{1/2} 1_{\verticesV}) (D^{1/2} 1_{\verticesV})^T}{\tr D} \right) D^{1/2}.
\]
Since $\tr D = \sum_{i=1}^{\verticesV} \omega_i = \sum_{i=1}^{\verticesV} (\sqrt{\omega_i})^2$, the inner matrix can be rewritten
\begin{equation}
I - \frac{(D^{1/2} 1_{\verticesV}) (D^{1/2} 1_{\verticesV})^T}{\tr D} = I - \frac{(D^{1/2} 1_{\verticesV}) (D^{1/2} 1_{\verticesV})^T}{\norm{D^{1/2} 1_{\verticesV}}^2},
\label{eq:orthoprojector}
\end{equation}
which makes it clear that this matrix is the orthogonal projector onto $(\Span D^{1/2} 1_{\verticesV})^\perp$. On the other hand, the graph Laplacian $L$ of the connected graph $\graphG$ has $\ker L = \Span(1_{\verticesV})$, so $\ker D^{-1/2} L D^{-1/2} = \Span (D^{1/2} 1_{\verticesV})$. Theorem~1 in~\cite{Bozzo} then states that 
\[
\mathcal{L}_\Omega^+ = \left( I - \frac{(D^{1/2} 1_{\verticesV}) (D^{1/2} 1_{\verticesV})^T}{\norm{D^{1/2} 1_{\verticesV}}^2} \right) D^{1/2} L^+ D^{1/2} \left( I - \frac{(D^{1/2} 1_{\verticesV}) (D^{1/2} 1_{\verticesV})^T}{\norm{D^{1/2} 1_{\verticesV}}^2} \right).
\]
Using the cyclic invariance of trace (and the fact that every projector is equal to its square) yields $\tr \mathcal{L}_{\Omega}^+ = \tr D \tr (QL^+)$, which completes the proof. 
\end{proof}

\section{The expected radius of gyration of subdivision graphs}

We are now ready to prove our main theorem, which we restate for convenience:

\begin{theorem}
%\label{thm:main}
Suppose $X$ is distributed as a Gaussian topological polymer (for $\graphG$) with edge variance $\sigma^2$, and that $\graphG$ is a subdivision of $\graphG'$, where each edge in $\graphG'$ has been subdivided into $n > 1$ pieces. Let $\Omega' = (\deg \vertex'_1 + \nicefrac{2}{n-1}, \dotsc, \deg \vertex'_{\verticesV'} + \nicefrac{2}{n-1})$, $\mathcal{L}_{\Omega'}(\graphG')$ be the weighted graph Laplacian of $\graphG'$, and $\operatorname{Loops}(\graphG') = \edgesE' - \verticesV' + 1$ be the cycle rank of $\graphG'$ (which is equal to the cycle rank of $\graphG$).

Then the expected (squared) radius of gyration
\begin{equation}
\left< \Rog(X) \right> = d \sigma^2 \frac{n^2-1}{2 \verticesV} \left( \frac{n}{n+1} \tr \mathcal{L}^+_{\Omega'}(\graphG') +  \frac{1}{3}\left(1 - \frac{1}{2\verticesV}\right) \operatorname{Loops}(\graphG') - \frac{1}{6}\left( 1 - \frac{1}{\verticesV} \right)\right). \label{eq:expected rog}
\end{equation}
The contraction factor of $\graphG$ is the quotient of this quantity by the corresponding expected (squared) radius of gyration for a linear polymer with the same number of monomers and the same edge variance $\sigma^2$. This is given by
\begin{equation}
g(\graphG) = \frac{3(n^2-1)}{\verticesV^2-1} \left( \frac{n}{n+1} \tr \mathcal{L}^+_{\Omega'}(\graphG') +  \frac{1}{3}\left(1 - \frac{1}{2\verticesV}\right) \operatorname{Loops}(\graphG') - \frac{1}{6}\left( 1 - \frac{1}{\verticesV} \right)\right). \label{eq:contraction factor}
\end{equation}
\end{theorem}

Since $\verticesV = (n-1) \edgesE' + \verticesV'$, both formulas are given entirely in terms of the original graph $\graphG'$ and the number of subdivisions $n$. 

\begin{proof}
Our goal is to compute the expectation of $\Rog(X)$, knowing that $X \dist \mathcal{N}(0,\sigma^2 L^+ \otimes I_d)$. Fix any $S$ in the group of permutation matrices $\mathcal{S}$ which preserve the edge groups $\mathcal{E}_i$. We showed in~\autoref{lem:invariance} that $X^S \dist \mathcal{N}(0,\sigma^2 L^+ \otimes I_d)$ as well. It follows that $\left< \Rog(X) \right> = \left< \Rog(X^S) \right>$. In particular, we may sum over all $S \in \mathcal{S}$ and (using linearity of expectations) obtain
\begin{equation}\label{eq:iterated expectations}
\left< \Rog(X) \right> = \frac{1}{\# \mathcal{S}} \sum_{S \in \mathcal{S}} \left< \Rog(X^S) \right> 
= \left< \frac{1}{\# \mathcal{S}} \sum_{S \in \mathcal{S}} \Rog(X^S) \right>.
\end{equation}
We note that this is a very simple application of the Law of Iterated Expectations. It is helpful because we have previously proved a theorem about these permutation averages: 

\begin{theorem}[{\cite[Theorem~1]{CantarellaSchumacherShonkwiler2024b}}]
\label{thm:symmetrization formula}
Suppose $X$ is an embedding of $\graphG$. For each $S \in \mathcal{S}$, let $X^S$ be the embedding induced by applying $S$ to the edge displacement vectors. Then the average radius of gyration over the $\mathcal{S}$-orbit of $X$ is
\begin{equation}
\begin{aligned}
	\frac{1}{\# \mathcal{S}} \sum_{S \in \mathcal{S}} \Rog(X^S) 
	&= 
	\Rog \pars[\Big]{ X',\deg + \frac{2}{n-1} }
	+ 
	\frac{n+1}{12 \ \verticesV^2} \pars{
	\pars{2 \verticesV - n} \norm{W}^2 
	-
	\pars{2 \verticesV - 1} \norm{W'}^2} .
\end{aligned} \label{eq:symmetrization theorem}
\end{equation}
Here $\Rog \pars[\Big]{X',\deg + \frac{2}{n-1} }$ is the weighted radius of gyration where $\vertex_i$ has weight $\deg(v_i) + \frac{2}{n-1}$, $\norm{W}^2 = \sum_{i=1}^\edgesE \sum_{j=1}^{n} \norm{w_{i,j}}^2$, and $\norm{W'}^2 = \sum_{i=1}^{\edgesE'} \norm{w'_i}^2$.
\end{theorem}

We now compute the expectation of the right-hand side of~\eqref{eq:symmetrization theorem}. \autoref{prop:embeddings of structure graph} says that if $X$ is distributed as a Gaussian topological polymer (for $\graphG$) with edge variance $\sigma^2$, then the translations of the corresponding embeddings $X'$ of $\graphG'$ which place their centers of mass at the origin are distributed as a Gaussian topological polymer (for $\graphG'$) with edge variance $n \sigma^2$. However, radius of gyration is invariant under translation. Thus the distribution of the random variable $\Rog(X',\deg + \frac{2}{n-1} )$ is the same as it would be if $X'$ were distributed as a Gaussian topological polymer (for $\graphG'$)  with edge variance $n \sigma^2$. It follows that from~\autoref{prop:weighted rog and weighted laplacian} that 
\begin{equation}
\begin{aligned}
\left<\Rog \pars[\Big]{ X',\deg + \frac{2}{n-1} }\right> &= \frac{d n \sigma^2}{\omega} \tr \mathcal{L}^+_{\Omega'}(\graphG') = \frac{d n (n-1) \sigma^2}{2 \verticesV} \tr \mathcal{L}^+_{\Omega'}(\graphG'),
\end{aligned}
\label{eq:expectation of rogGp}
\end{equation}
since
\begin{equation}
\omega = \sum_{i=1}^{\verticesV'} \omega_i =  \sum_{i=1}^{\verticesV'}\left( \deg(\vertex'_i) + \frac{2}{n-1}\right) = 
2 \edgesE' + \frac{2 \verticesV'}{n-1} = \frac{2 \pars{(n-1) \edgesE' + \verticesV'}}{n-1} = \frac{2\verticesV}{n-1}.  \label{eq:total weight}
\end{equation}
Now we know from~\autoref{defn:phantom network distribution} that $\Vec W \dist \mathcal{N}(0,\sigma^2 P)$. Thus,
\[
\left<\norm{W}^2\right> = \tr(\sigma^2 P) = \sigma^2 \dim (\im P) = \sigma^2 d(\edgesE - \operatorname{Loops}(\graphG)) = d \sigma^2 (\edgesE - \operatorname{Loops}(\graphG')),
\]
recalling that $P$ is an orthogonal projector (so its trace is equal to the dimension of its image) and using the fact that $\im P = \im \bdy^T \otimes I_d$ along with~\autoref{prop:ker and image described}.\footnote{We note that this line is essentially Foster's proof of his celebrated theorem~\cite{Foster1948} on the average of pairwise resistance distances.}

Again using~\autoref{prop:embeddings of structure graph}, we know $\Vec W' \dist \mathcal{N}(0,n \sigma^2 P')$, so as above, we have
\[
\left<\norm{W'}^2\right> = d \sigma^2 n (\edgesE' - \operatorname{Loops}(\graphG')).
\]
Using $\chi = \operatorname{Loops}(\graphG) = \operatorname{Loops}(\graphG')$ for brevity and observing that $n \edgesE' = \edgesE = \verticesV + \chi - 1$, we now have
\begin{equation}
\begin{aligned}
\left<
	\pars{2 \verticesV - n} \norm{W}^2 
	-
	\pars{2 \verticesV - 1} \norm{W'}^2\right> &= 
(2 \verticesV - n)(d \sigma^2)(\edgesE - \chi) - (2 \verticesV - 1)(d \sigma^2 n)(\edgesE' - \chi) \\
&= d \sigma^2 (n-1) (2 \verticesV \chi - n \edgesE') \\
&= d \sigma^2 (n-1) \pars[Big]{ (2 \verticesV - 1) \chi - (\verticesV - 1) }.
\end{aligned}
\label{eq:expectation of other stuff}
\end{equation}
Combining~\eqref{eq:iterated expectations} with~\eqref{eq:symmetrization theorem},~\eqref{eq:expectation of rogGp}, and~\eqref{eq:expectation of other stuff} yields~\eqref{eq:expected rog}. To get~\eqref{eq:contraction factor}, we observe that the expected radius of gyration of a $\verticesV$ vertex linear polymer whose embeddings are distributed in $\R^d$ as a Gaussian topological polymer with edge variance $\sigma^2$ is known to be $d \sigma^2 \frac{\verticesV^2-1}{6 \verticesV}$ (cf.~\cite[eq.~(8)]{Cantarella2022ROG}).
\end{proof}
%If we need a proof for the linear polymer, it's here:
%\begin{proof}
%The average over rearrangements of $\Rog$ is given by $\frac{\verticesV+1}{12\verticesV}(\sum_{i=1}^\edgesE \norm{w_i}^2 + \norm{\sum_{i=1}^\edgesE w_i}^2)$, as noted above (cf. Proposition~6.5 in~\cite{Cantarella2012d}). The expectation of $\norm{w_i}^2$ is $\sigma^2$ for each $i$. Further, since the $w_i$ are independent, the expectation of $\norm{\sum_{i=1}^\edgesE w_i}^2$ is $\edgesE \sigma^2$.
%\end{proof} 

As we noted above, we can rewrite~\autoref{thm:main} in terms of the Kirchhoff index and a weighted Kirchhoff index, obtaining~\autoref{cor:main cor}.

\section{Example computations}

This formula can make computing the exact contraction factor of a subdivided graph rather simple, especially with the aid of the following helpful Lemma. 

\begin{lemma} 
\label{lem:computations}
If $\graphG'$ is a graph with $\verticesV'$ vertices, suppose $\Omega'$ are weights given by $\omega'_i = \deg v'_i + \frac{2}{n-1}$ and $D = \diag(\omega_1, \dotsc, \omega_{\verticesV'})$ is the diagonal matrix of weights. Then if $L(\graphG')$ is the graph Laplacian, define $p(\lambda) = \det(L(\graphG') - \lambda D) =   c_{\verticesV'} \lambda^{\verticesV'} + \cdots + c_1 \lambda + c_0$. Then $\tr \mathcal{L}^+_{\Omega'}(\graphG') = -\frac{c_2}{c_1}$.
\end{lemma}

\begin{proof}
We know for any matrix $A$ with nullity $1$ that $\tr A^+ = -\frac{c_2}{c_1}$ where these are coefficients of the characteristic polynomial. Since $\graphG'$ is a connected graph, the graph Laplacian $L$ has nullity $1$. Since $\mathcal{L}_{\Omega'} = D^{-1/2} L D^{-1/2}$ is similar to $L$, it has the same nullity. Further, we can write its characteristic polynomial as 
\[
q(\lambda) = \det(D^{-1/2} L D^{-1/2} - \lambda I) = \det(D^{-1/2} (L - \lambda D) D^{-1/2}) = \det(D^{-1}) \det(L - \lambda D).
\]
We see that $q(\lambda) = \det(D^{-1}) p(\lambda)$, so the quotient of coefficients in $q(\lambda)$ agrees with the corresponding quotient in $p(\lambda)$.
\end{proof}

\begin{figure}[t]
\hfill
\begin{overpic}[height=0.25\textwidth]{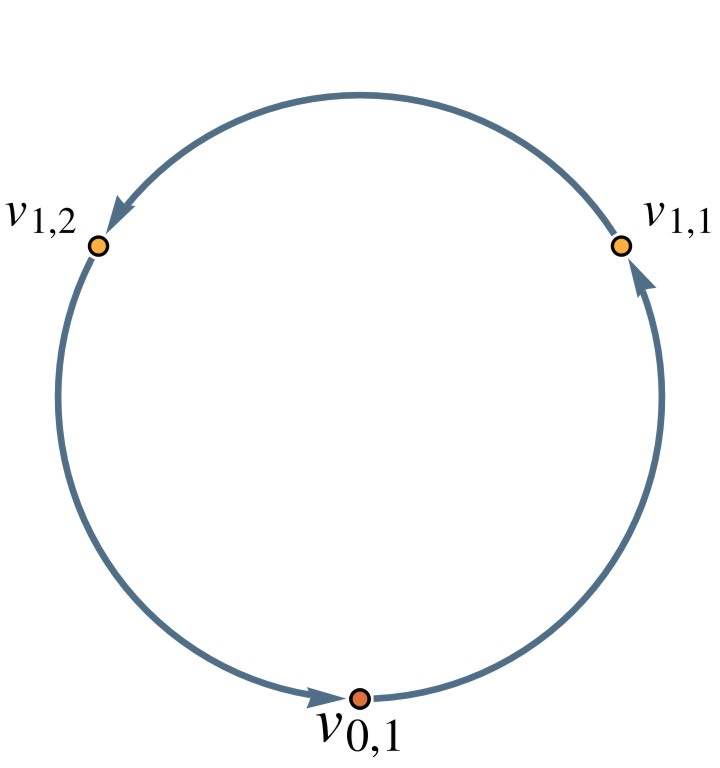} 
\end{overpic}
\hfill
\begin{overpic}[height=0.25\textwidth]{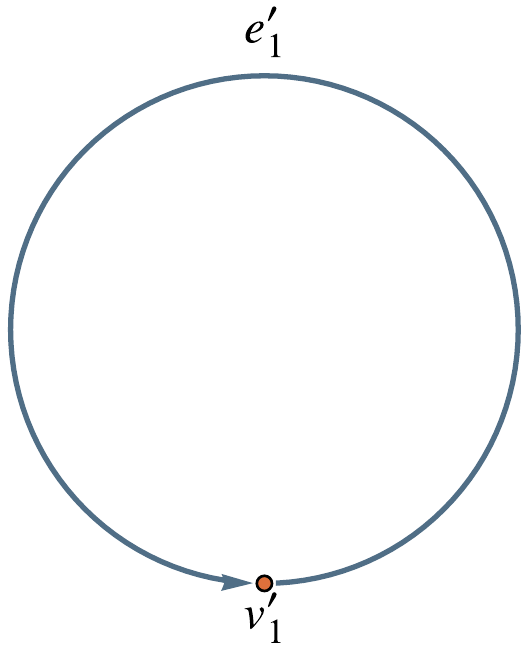} 
\end{overpic}
\hfill
\hphantom{.}
\caption{The cycle graph $\graphG$ (at left) is a subdivision of the multigraph $\graphG'$ with one vertex and one edge (at right).}
\label{fig:cycle}
\end{figure}

We first consider the cycle graph, as illustrated in \autoref{fig:cycle}. In this case, the structure graph $\graphG'$ has one vertex and one edge. The graph Laplacian $L = [0]$, so the weighted graph Laplacian $\mathcal{L}_\Omega' = D^{-1/2} L D^{-1/2} = [0]$ as well. It follows that $\tr \mathcal{L}^+_{\Omega'} = 0$. The cycle graph has one loop and $n$ vertices, so~\eqref{eq:contraction factor} becomes
\[
g(\graphG) = \frac{3(n^2-1)}{n^2-1} \left(\frac{1}{3}\left(1 - \frac{1}{2n}\right) - \frac{1}{6}\left( 1 - \frac{1}{n} \right)\right) = \frac{1}{2},
\]
as expected. 

\begin{figure}[t]
\hfill
\begin{overpic}[width=0.3\textwidth]{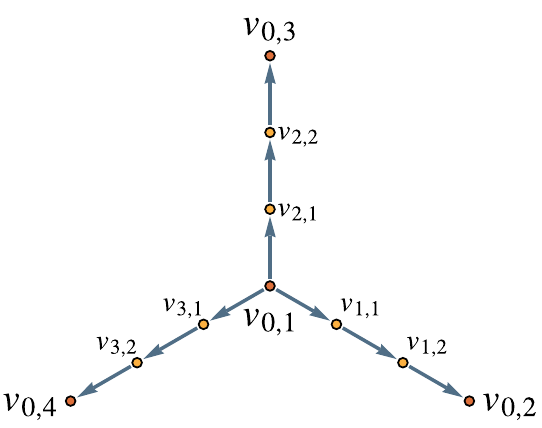} 
\end{overpic}
\hfill
\begin{overpic}[width=0.3\textwidth]{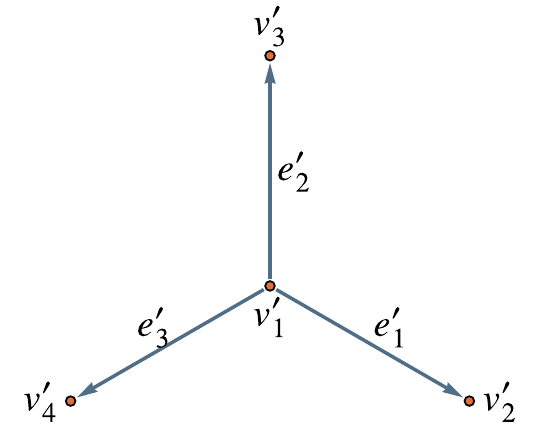} 
\end{overpic}
\hfill
\hphantom{.}
\caption{The star polymer graph $\graphG$ (at left) is a subdivision of the star graph with four vertices $\graphG'$ (at right).}
\label{fig:ygraph}
\end{figure}

The contraction factor of a star polymer with $f$ branches (in the limit as the number of subdivisions $n \rightarrow \infty$) is given by $(3f - 2)/f^2$~\cite{Teraoka2002}. Consider the case $f=3$ (Y-shaped star polymer), as shown in \autoref{fig:ygraph}, where we expect an asymptotic limit of $7/9$.  In this case, the structure graph $\graphG'$ is the 4-vertex star graph with a central vertex of degree 3 and 3 vertices of degree 1, so we have
\[
\begin{aligned}
\det(L - \lambda D) &= 
\begin{vsmallmatrix}
 3-\lambda  \left(3 + \frac{2}{n-1}\right) & -1 & -1 & -1 \\
 -1 & 1-\lambda  \left(1+ \frac{2}{n-1}\right) & 0 & 0 \\
 -1 & 0 & 1-\lambda  \left(1+\frac{2}{n-1}\right) & 0 \\
 -1 & 0 & 0 & 1-\lambda  \left(1+\frac{2}{n-1}\right) 
\end{vsmallmatrix} \\
&= -\frac{2  (3
   n+1)}{n-1} \lambda +  \frac{3 \left(5 n^2+6
   n+1\right)}{(n-1)^2} \lambda ^2 + c_3 \lambda^3 + c_4 \lambda^4.
\end{aligned}
\]
where the coefficients $c_3$ and $c_4$ are irrelevant to our computation, so we don't write them out. It now follows from~\autoref{lem:computations} that 
\[
\tr \mathcal{L}^+_{\Omega'}(\graphG') = \frac{15 n^2+18 n+3}{6 n^2-4 n-2}.
\]
There are no loops in the star polymer, and the number of vertices $\verticesV = 3n+1$. This means that the contraction factor of the subdivided graph, according to \eqref{eq:contraction factor}, becomes
\[
\begin{aligned}
g(\graphG) &= \frac{7n^2 + 9n + 2}{9n^2 + 9n + 2} = \frac{7}{9}+\frac{2}{9 n}-\frac{14}{81 n^2}+\frac{10}{81
   n^3}+O\left(\frac{1}{n^4}\right).
\end{aligned}
\]
This matches an exact formula for the Y-shaped star polymer with $n$ subdivisions given by Bonchev et al.~\cite[eqs.~(37) and~(39)]{Bonchev:2002vy}. 

\begin{figure}[t]
\hfill
\begin{overpic}[width=0.3\textwidth]{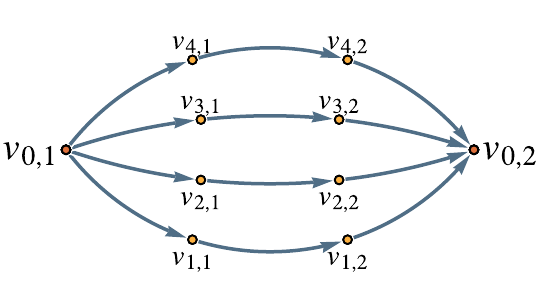} 
\end{overpic}
\hfill
\begin{overpic}[width=0.3\textwidth]{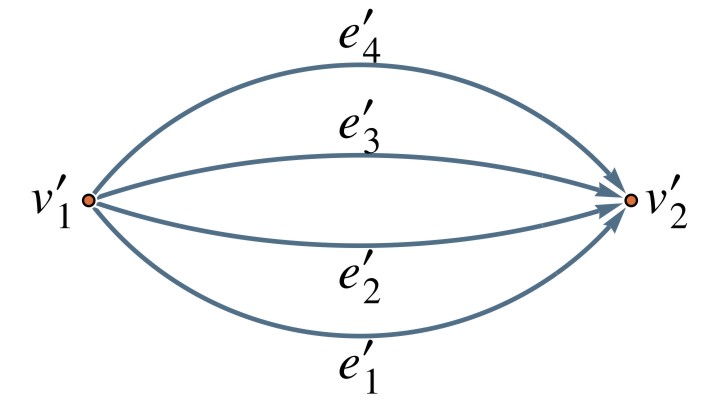} 
\end{overpic}
\hfill
\hphantom{.}
\caption{The graph $\graphG$ (at left) is a subdivision of the multitheta graph with 4 edges $\graphG'$ (at right).}
\label{fig:theta}
\end{figure}

We next consider the subdivided multitheta graph, as shown in \autoref{fig:theta}. Here $\graphG'$ is the graph with $m$ edges joining two vertices, so $\verticesV' = 2$, and $\edgesE' = m$. Both vertices have degree $m$, so we compute
\[
\det(L - \lambda D) = 
\begin{vsmallmatrix} m-\lambda  \left(m+\frac{2}{n-1}\right) & -m \\
 -m & m-\lambda  \left(m+\frac{2}{n-1}\right) \\
\end{vsmallmatrix} = \frac{(m(n-1)+2)^2}{(n-1)^2}\lambda ^2  - \frac{2m(m(n-1)+2)}{n-1}\lambda.
\]
It follows from~\autoref{lem:computations} that $\tr \mathcal{L}^+_{\Omega'}(\graphG') = \frac{1}{m (n-1)}+\frac{1}{2}$, and then from~\autoref{thm:main} that the contraction factor for the $n$-subdivided $m$-multitheta graph is given by 
\[
\begin{aligned}
g(\theta_{m,n}) &= \frac{m^3 n^3 -((m-2) m (m+2)) n^2 + ((m-2) (m-1) m) n -m^3+3 m^2-6 m+6} 
{m^4 n^3 -(3 (m-2) m^3) n^2 + m^2 (3 m^2-12 m+11) n - (m-3) (m-2) (m-1) m} \\
   &= \frac{1}{m}+\frac{2 (m-1)(m-2)}{m^3
   n}+O\left(\frac{1}{n^2}\right).
\end{aligned}
\]
We see that the large-$n$ limit of the contraction factor is $1/m$, in agreement with~\cite[eq.~3]{Uehara:2018bb}. The correction for finite $n$ is quite small: for the $\theta_{3,4}$ graph, the contraction factor may be computed explicitly as $\frac{241}{660} \sim 0.365$ which is very close to $\frac{1}{3}$. 

\begin{figure}[t]
\hfill
\begin{overpic}[width=0.4\textwidth]{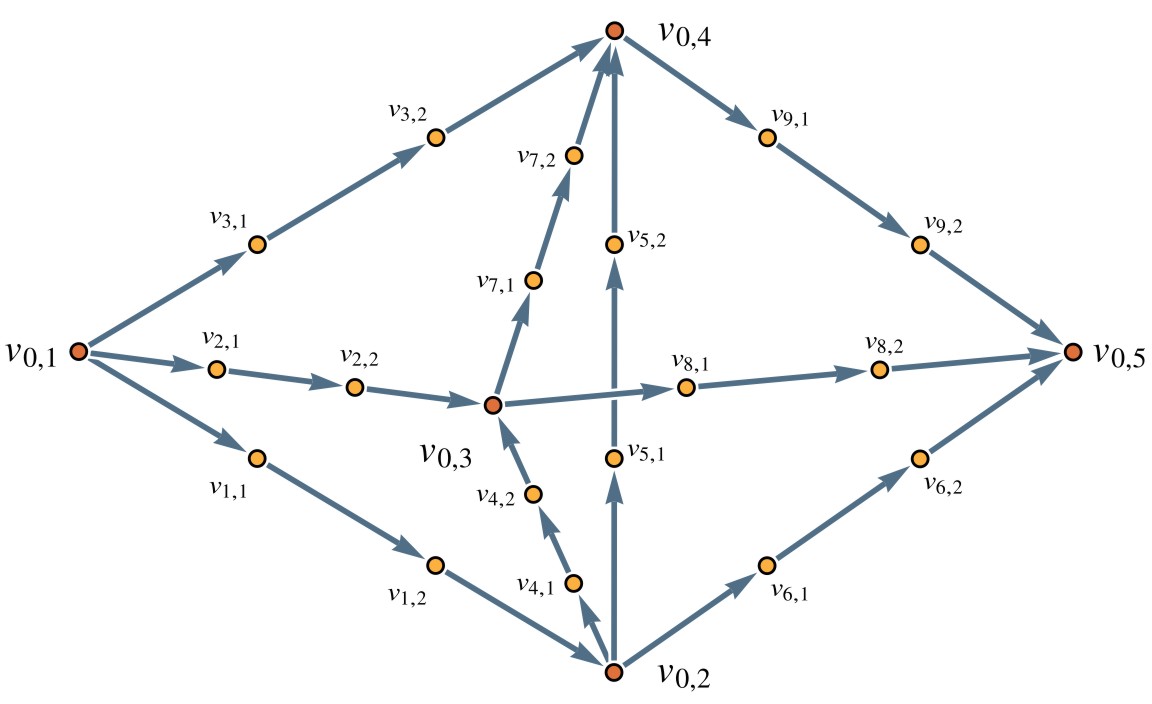} 
\end{overpic}
\hfill
\begin{overpic}[width=0.4\textwidth]{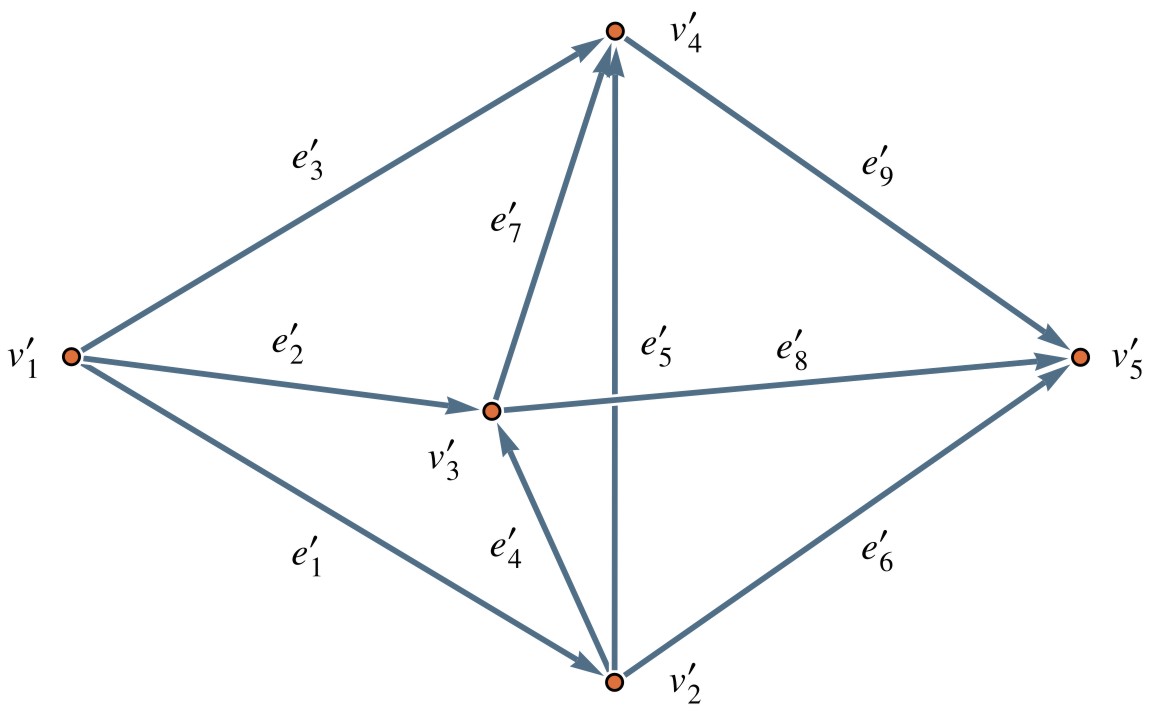} 
\end{overpic}
\hfill
\hphantom{.}
\caption{The graph $\graphG$ (at left) is a subdivision of the bipyramid graph $\graphG'$ (at right).}
\label{fig:example}
\end{figure}

Finally, we give an example of applying this method to a new graph that we have not found in the literature. Consider the 5 vertex bipyramid graph $\graphG'$ in~\autoref{fig:example}, where three of the vertices have degree 4 and the other two have degree 3, and an arbitrary subdivision $\graphG$ where each edge of $\graphG'$ has been subdivided into $n$ edges. We have
\[
\det(L - \lambda D) =  
\begin{vsmallmatrix}
 4- \lambda  \left(4 + \frac{2}{n-1}\right) & -1 & -1 & -1  & -1 \\
 -1 & 4- \lambda  \left(4 + \frac{2}{n-1}\right) & -1 & -1 & -1 \\
 -1 & -1 & 4-  \lambda  \left(4 + \frac{2}{n-1}\right) & -1 & -1 \\
 -1 & -1 & -1 & 3- \lambda  \left(3 + \frac{2}{n-1}\right)  & 0 \\
 -1 & -1 & -1 & 0 & 3- \lambda \left(3+\frac{2}{n-1}\right) \\
\end{vsmallmatrix}.
\]
Evaluating the determinant yields
\[
p(\lambda) = c_5 \lambda^5 + \cdots +    
   \frac{10  \left(441 n^2-384
   n+83\right)}{(n-1)^2} \lambda ^2
   +
   \frac{150 (4-9
   n)}{n-1} \lambda  
\]
and this implies that $\tr \mathcal{L}_{\Omega'}(\graphG') = (49/15) + 28/(15 (-1 + n)) + 1/(15 (-4 + 9 n))$.
Since $\graphG'$ has $9$ edges, we know $\operatorname{Loops}(\graphG') = \edgesE' - \verticesV' + 1 = 9 - 5 + 1 = 5$. Further, for the subdivided graph, $\verticesV = \verticesV' + (n-1) \edgesE' = 5 + 9(n-1)$. Plugging this all into~\eqref{eq:contraction factor} yields an exact formula for $g(\graphG)$:
\[
g(\graphG) = \frac{1287 n^3-1043 n^2-239 n+275}{30 (9 n-4) \left(27
   n^2-24 n+5\right)} = \frac{143}{810}+\frac{673}{7290
   n}+O\left(\frac{1}{n^2} \right).
\]
%This is actually done in subdivided-phantom-network-example.nb

\section*{Acknowledgments}
We are very grateful for the support of the National Science Foundation (DMS--2107700 to Shonkwiler) and the Simons Foundation (\#524120 to Cantarella), as well as many helpful discussions with colleagues about radius of gyration and Gaussian topological polymers, including Henrik Schumacher, Sebastian Caillaut and Azim Sharapov. This work was partially supported by the Japan Science and Technology Agency CREST: Grant Number JPMJCR19T4. 

\bibliography{papers-special,mendeley-export}

\end{document}

%% file: Macros.tex
\newcommand{\Kirchhoff}{\operatorname{Kf}}
\newcommand{\Rog}{\operatorname{R}_{\mathrm{g}}^2}

\newcommand{\R}{\mathbb{R}}

\newcommand{\Hom}{\operatorname{Hom}}

\newcommand{\bdry}{\partial}
\newcommand{\bdy}{\bdry}

\newcommand{\diag}{\operatorname{diag}}

\newcommand{\Loops}{\operatorname{Loops}}

\newcommand{\Mat}{\operatorname{Mat}}

\newcommand{\from}{\co\!\!}

\newcommand{\tr}{\operatorname{tr}}

\newcommand{\graphG}{\mathbf{G}}
\newcommand{\edgesE}{\mathbf{e}}
\newcommand{\verticesV}{\mathbf{v}}
\newcommand{\edge}{e}
\newcommand{\vertex}{v}
\newcommand{\head}{\operatorname{head}}
\newcommand{\tail}{\operatorname{tail}}

\newcommand{\im}{\operatorname{im}}

\newcommand{\Span}{\operatorname{span}}
\newcommand{\dist}{\sim}

\newcommand{\ceq}{\coloneqq}

\def\co{\colon\thinspace}

\let\mgp=\marginpar \marginparwidth18mm \marginparsep1mm
\def\marginpar#1{\mgp{\raggedright\tiny #1}}

\let\lbl=\label
\def\label#1{\lbl{#1}\ifinner\else\marginpar{\ref{#1} #1}\ignorespaces\fi}

\DeclarePairedDelimiterXPP{\pars}[1]{\mathop{}}{\lparen}{\rparen}{}{#1}

\DeclarePairedDelimiterXPP{\abs}[1]{\mathop{}}{\lvert}{\rvert}{}{#1}
\DeclarePairedDelimiterXPP{\norm}[1]{\mathop{}}{\lVert}{\rVert}{}{#1}
\DeclarePairedDelimiterXPP{\seminorm}[1]{\mathop{}}{\lbrack}{\rbrack}{}{#1}
\DeclarePairedDelimiterXPP{\inner}[1]{\mathop{}}{\langle}{\rangle}{}{#1}
\DeclarePairedDelimiterXPP{\brackets}[1]{\mathop{}}{\lbrack}{\rbrack}{}{#1}
\DeclarePairedDelimiterXPP{\braces}[1]{\mathop{}}{\lbrace}{\rbrace}{}{#1}

\DeclarePairedDelimiterXPP{\intervalcc}[1]{\mathop{}}{\lbrace}{\rbrace}{}{#1}
\DeclarePairedDelimiterXPP{\intervalco}[1]{\mathop{}}{\lbrace}{\rparen}{}{#1}
\DeclarePairedDelimiterXPP{\intervaloc}[1]{\mathop{}}{\lparen}{\rbrace}{}{#1}
\DeclarePairedDelimiterXPP{\intervaloo}[1]{\mathop{}}{\lparen}{\rparen}{}{#1}

\DeclarePairedDelimiterXPP{\set}[2]{\mathop{}}{\lbrace}{\rbrace}{}{#1\,\delimsize\vert\,\mathopen{}#2}
\let\dot\undefined
\DeclarePairedDelimiterXPP{\dot}[2]{\mathop{}}{\langle}{\rangle}{}{#1,#2}
\DeclarePairedDelimiterXPP{\floor}[1]{\mathop{}}{\lfloor}{\rfloor}{}{#1}
\DeclarePairedDelimiterXPP{\ceil}[1]{\mathop{}}{\lceil}{\rceil}{}{#1}

\DeclareDocumentCommand{\converges}{ o }{
	\mathbin{%
		\IfValueTF{#1}{%
			\mathrel{\vbox{\offinterlineskip\ialign{%
				\hfil##\hfil\cr
				$\scriptscriptstyle#1$\cr
				$-\!\!\!-\!\!\!\rightarrow$\cr
			}}}
		}{%
			-\!\!\!-\!\!\!\rightarrow
		}%
	}%
}

%% file: contractionfactors.bbl
\begin{thebibliography}{10}

\bibitem{babic_resistance-distance_2002}
Darko Babi\'{c}, Douglas~J. Klein, Istv\'{a}n Lukovits, Sonja Nikoli\'{c}, and
  Nenad Trinajsti\'{c}.
\newblock Resistance-distance matrix: A computational algorithm and its
  application.
\newblock {\em International Journal of Quantum Chemistry}, 90:166--176, 2002.

\bibitem{Bapat:cx}
Ravindra~B. Bapat, Ivan Gutman, and Wenjun Xiao.
\newblock A simple method for computing resistance distance.
\newblock {\em Zeitschrift f\"{u}r Naturforschung}, pages 494--498, 2003.

\bibitem{bonchev_molecular_1994}
Danail Bonchev, Alexandru~T. Balaban, Xiaoyu Liu, and Douglas~J. Klein.
\newblock Molecular cyclicity and centricity of polycyclic graphs. {I}.
  {C}yclicity based on resistance distances or reciprocal distances.
\newblock {\em International Journal of Quantum Chemistry}, 50:1--20, 1994.

\bibitem{Bonchev:2002vy}
Danail Bonchev, Eric~J Markel, and Armenag~H Dekmezian.
\newblock Long chain branch polymer chain dimensions: application of topology
  to the {Z}imm{--}{S}tockmayer model.
\newblock {\em Polymer}, 43:203--222, 2002.

\bibitem{Bozzo}
Enrico Bozzo.
\newblock The {M}oore--{P}enrose inverse of the normalized graph {L}aplacian.
\newblock {\em Linear Algebra and Its Applications}, 439:3038--3043, 2013.

\bibitem{Cantarella2022ROG}
Jason Cantarella, Tetsuo Deguchi, Clayton Shonkwiler, and Erica Uehara.
\newblock Radius of gyration, contraction factors, and subdivisions of
  topological polymers.
\newblock {\em Journal of Physics A: Mathematical and Theoretical}, 55:475202,
  2022.

\bibitem{CantarellaSchumacherShonkwiler2024b}
Jason Cantarella, Henrik Schumacher, and Clayton Shonkwiler.
\newblock On the average squared radius of gyration of a family of embeddings
  of subdivision graphs.
\newblock Preprint, {\tt arXiv:2409.18767 [math.CO]}, 2024.

\bibitem{Carmona:2017ge}
\'{A}ngeles Carmona, Margarida Mitjana, and Enric Mons\'{o}.
\newblock Effective resistances and {K}irchhoff index in subdivision networks.
\newblock {\em Linear and Multilinear Algebra}, 65:1823--1837, 2017.

\bibitem{Chen:2010da}
Haiyan Chen.
\newblock Random walks and the effective resistance sum rules.
\newblock {\em Discrete Applied Mathematics}, 158:1691--1700, 2010.

\bibitem{Eichinger1980}
Bruce~E. Eichinger.
\newblock Configuration statistics of {G}aussian molecules.
\newblock {\em Macromolecules}, 13:1--11, 1980.

\bibitem{Eichinger1978:reduction}
Bruce~E. Eichinger and J.~E. Martin.
\newblock Distribution functions for {G}aussian molecules. {II}. reduction of
  the {K}irchhoff matrix for large molecules.
\newblock {\em The Journal of Chemical Physics}, 69:4595--4599, 1978.

\bibitem{FloryPaulJ1969Smoc}
Paul~J. Flory.
\newblock {\em Statistical {M}echanics of {C}hain {M}olecules.}
\newblock Interscience Publishers, New York, 1969.

\bibitem{Flory1976}
Paul~J. Flory.
\newblock Statistical thermodynamics of random networks.
\newblock {\em Proceedings of the Royal Society of London. Series A,
  Mathematical and Physical Sciences}, 351:351--380, 1976.

\bibitem{Foster1948}
Ronald~M Foster.
\newblock The average impedance of an electrical network.
\newblock {\em Contributions to Applied Mechanics (Reissner Anniversary
  Volume)}, pages 333--340, 1949.

\bibitem{Gao:2012jk}
Xing Gao, Yanfeng Luo, and Wenwen Liu.
\newblock Kirchhoff index in line, subdivision and total graphs of a regular
  graph.
\newblock {\em Discrete Applied Mathematics}, 160:560--565, 2012.

\bibitem{Gutman:1996hq}
Ivan Gutman and Bojan Mohar.
\newblock The quasi-{W}iener and the {K}irchhoff indices coincide.
\newblock {\em Journal of Chemical Information and Computer Sciences},
  36:982--985, 1996.

\bibitem{HatcherAT}
Allen~E. Hatcher.
\newblock {\em Algebraic Topology}.
\newblock Cambridge University Press, Cambridge, 2002.

\bibitem{James:1947hp}
Hubert~M. James.
\newblock Statistical properties of networks of flexible chains.
\newblock {\em The Journal of Chemical Physics}, 15:651--668, 1947.

\bibitem{James1943}
Hubert~M. James and Eugene Guth.
\newblock Theory of the elastic properties of rubber.
\newblock {\em The Journal of Chemical Physics}, 11:455--481, 1943.

\bibitem{Klein:2002vx}
Douglas~J. Klein.
\newblock Resistance-distance sum rules.
\newblock {\em Croatica Chemica Acta}, 75:633--649, 2002.

\bibitem{klein_random_2004}
Douglas~J. Klein, Jos\'{e}~Luis Palacios, Milan Randi\'{c}, and Nenad
  Trinajsti\'{c}.
\newblock Random walks and chemical graph theory.
\newblock {\em Journal of Chemical Information and Computer Sciences},
  44:1521--1525, 2004.

\bibitem{Klein:1993tb}
Douglas~J. Klein and Milan Randi\'{c}.
\newblock Resistance distance.
\newblock {\em Journal of Mathematical Chemistry}, 12:81--95, 1993.

\bibitem{kuratafukatsu1964}
Michio Kurata and Masaaki Fukatsu.
\newblock Unperturbed dimension and translational friction constant of branched
  polymers.
\newblock {\em The Journal of Chemical Physics}, 41:2934--2944, 1964.

\bibitem{MasseyBasicCourse}
William~S. Massey.
\newblock {\em A Basic Course in Algebraic Topology}.
\newblock Springer-Verlag, New York, 1991.

\bibitem{RomanAdvancedLinearAlgebra}
Steven Roman.
\newblock {\em Advanced Linear Algebra}.
\newblock Springer-Verlag, New York, 3rd edition, 2008.

\bibitem{stepto:2015be}
Robert Stepto, Taihyun Chang, Pavel Kratochv\'{\i}l, Michael Hess, Kazuyuki
  Horie, Takahiro Sato, and Ji\v{r}\'{\i} Vohl\'{\i}dal.
\newblock Definitions of terms relating to individual macromolecules,
  macromolecular assemblies, polymer solutions, and amorphous bulk polymers
  ({IUPAC} recommendations 2014).
\newblock {\em Pure and Applied Chemistry}, 87:71--120, 2015.

\bibitem{Teraoka2002}
Iwao Teraoka.
\newblock {\em Polymer Solutions: An Introduction to Physical Properties}.
\newblock John Wiley and Sons, 2002.

\bibitem{Tezuka:2017gh}
Yasuyuki Tezuka.
\newblock Topological polymer chemistry designing complex macromolecular graph
  constructions.
\newblock {\em Accounts of Chemical Research}, 50:2661--2672, 2017.

\bibitem{Uehara:2018bb}
Erica Uehara and Tetsuo Deguchi.
\newblock Statistical properties of multi-theta polymer chains.
\newblock {\em J. Phys. A}, 51:134001--134022, 2018.

\bibitem{Yang2016}
Yujun Yang.
\newblock The {K}irchhoff index of subdivisions of graphs.
\newblock {\em Discrete Applied Mathematics}, 171:153--157, 2014.

\bibitem{Yang:2015eq}
Yujun Yang and Douglas~J. Klein.
\newblock Resistance distance-based graph invariants of subdivisions and
  triangulations of graphs.
\newblock {\em Discrete Applied Mathematics}, 181:260--274, 2015.

\bibitem{Zhang:2006bm}
Heping Zhang and Yujun Yang.
\newblock Resistance distance and {K}irchhoff index in circulant graphs.
\newblock {\em International Journal of Quantum Chemistry}, 107:330--339, 2007.

\bibitem{Zimm1949}
Bruno~H. Zimm and Walter~H. Stockmayer.
\newblock The dimensions of chain molecules containing branches and rings.
\newblock {\em The Journal of Chemical Physics}, 17:1301, 1949.

\end{thebibliography}
